\documentclass[journal,twoside,web]{ieeecolor}

\usepackage{generic}
\usepackage{cite}
\usepackage{amsmath,amssymb,amsfonts}
\usepackage{algorithmic}
\usepackage{graphicx}
\usepackage{textcomp}

\usepackage{graphicx}
\usepackage{subfigure}
\usepackage{epsfig} 
\usepackage{cite}
\usepackage{lcsys}
\usepackage{color}
\usepackage{url}

\newtheorem{theorem}{Theorem}[section]

\newtheorem{proposition}[theorem]{Proposition}

\newtheorem{remark}[theorem]{Remark}
\newtheorem{assumption}[theorem]{Assumption}

\begin{document}

\def\BibTeX{{\rm B\kern-.05em{\sc i\kern-.025em b}\kern-.08em
    T\kern-.1667em\lower.7ex\hbox{E}\kern-.125emX}}
\markboth{\journalname, VOL. XX, NO. XX, XXXX 2017}
{Author \MakeLowercase{\textit{et al.}}: Preparation of Papers for IEEE Control Systems Letters (August 2022)}

\title{Leader-Follower Formation Control of Perturbed Nonholonomic Agents along Parametric Curves with Directed Communication}

\author{
Bin Zhang, Hui Zhi, Jose Guadalupe Romero, and David Navarro-Alarcon%
\thanks{This work is supported by the Research Grants Council (RGC) of Hong Kong under grant 15212721 and grant 15231023. \textit{(Corresponding author: D. Navarro-Alarcon.)}}
\thanks{B. Zhang, H. Zhi, and D. Navarro-Alarcon are with The Hong Kong Polytechnic University (PolyU), Kowloon, Hong Kong. (e-mail: me-bin.zhang@connect.polyu.hk; hui1225.zhi@connect.polyu.hk; dnavar@polyu.edu.hk)}
\thanks{J. G. Romero is with the Instituto Tecnológico Autónomo de México (ITAM), Mexico City, Mexico. (e-mail: jose.romerovelazquez@itam.mx)}
}

\maketitle
\thispagestyle{empty}

\begin{abstract}
In this paper, we propose a novel formation controller for nonholonomic agents to form general parametric curves.
First, we derive a unified parametric representation for both open and closed curves.
Then, a leader-follower formation controller is designed to form the parametric curves.
We consider directed communications and constant input disturbances rejection in the controller design.
Rigorous Lyapunov-based stability analysis proves the asymptotic stability of the proposed controller.
Detailed numerical simulations and experimental studies are conducted to verify the performance of the proposed method.
\end{abstract}


\section{Introduction}
\label{sec:introduction}

\IEEEPARstart{T}{he} application of multi-agent systems has gained popularity over the past decades\cite{olfati2007consensus}.
A typical multi-agent system usually has multiple layers to deal with assigned tasks, such as formation control, path planning, environmental sensing, and coordination\cite{ismail2018survey}.
Among those, multi-agent formation control is the fundamental problem studied by many researchers.

Formation control technology can be used in various areas, such as object transportation\cite{alonso2017multi}, rescue\cite{chen2005formation}, and environmental surveillance\cite{lu2016cooperative}.
In those scenarios, agents are usually controlled to form a desired pattern and keep or deform it following the task requirements\cite{lo2022formation}.
Generally, agents can achieve formation control in three ways, that is, position-based, displacement-based, and distance-based methods\cite{oh2015survey}.
The platforms that carry out the formation control can be mobile robots\cite{diaz2023leader, izumi2022formation}, satellites\cite{scharf2004survey}, manipulators\cite{wu2022distributed}, etc.
In the previous literature, the desired patterns that agents need to form are usually specified by some absolute or relative positions, and those patterns are usually simple in terms of morphology, such as polygons and lines.
This could cause problems when dealing with complicated cases where the desired pattern has a complex structure, such as curves, because it is usually difficult to fully depict the characteristics of such complex patterns simply with several positions.
There are also scenarios where we only care about the overall distribution rather than the specific positions of agents, such as caging the object in the object transportation tasks and monitoring the boundaries of spilled oil. 
In that case, representing the desired pattern with several positions could limit the adaptivity and flexibility of the agents.

In this work, we consider adopting a more general way, the parametric curves,  to represent the desired pattern the agents need to form and develop a leader-follower formation control method to drive agents to form those parametric curves.
In this way, our method has high flexibility to deal with formation tasks with complex desired patterns, such as open and closed curves with complex morphology.
Representing the desired patterns with parametric equations rather than positions also endows our method with high scalability, since we specify the target of agents by parameters that do not require measurements.
We also consider the requirements for real-world applications.
The developed method can work under directed communications and reject constant input disturbances.

\emph{Related work.}
Compared with the abundant literature on multi-agent formation control, the works on driving agents to form complex patterns or parametric curves are limited.
Early work can be seen in \cite{zhang2007coordinated}, where the author developed a steering control to drive particles to travel along closed invariant patterns (curves). 
In this work, particles are coupled in a chain structure and limited to a constant unit speed.
Hsieh \emph{et. al.}\cite{hsieh2008decentralized} proposed a decentralized controller for robotic swarms to form closed planar curves.
In this work, agents with double-integrator dynamics were considered, and they did not need to exchange information during the control process.
Song \emph{et. al.}\cite{cheng2021finite} studied the coverage control of a multi-agent system on a closed curve.
They considered the single-integrator agents and optimized the coverage cost of agents in a finite time.
A Fourier series-based method was developed in \cite{zhang2023fourier} for nonholonomic agents to track evolving curves, but a strongly connected communication topology was required.
More examples can also be seen in \cite{saldana2021estimating,wang2019leaderless}, where estimating and tracking environmental boundaries, such as profiles of plumes and liquids, by multi-agent systems are studied.

Reviewing the previous works, we can find that most works considered closed or open curves, but seldom considered both in a single work.
Also, most previous works did not consider the disturbance rejection in their methods, which significantly limits their applications in real-world missions.
Besides, most previous works only considered simple agents' dynamics, such as single-integrator and double-integrator dynamics. 
Extensions to more complicated dynamics, such as nonholonomic dynamics, still need to be considered.
More importantly, experiments on real-world systems should be conducted to verify the performance of the proposed methods, which was not done by many previous works.

\emph{Our contribution.}
In this work, we study the formation control problem of multi-agent systems to form complicated patterns represented by parametric curves.
Compared with the normal formation control method, our proposed method specifies agents' targets by curve parameters rather than specific absolute or relative positions, thus having high scalability.
Compared with previous related work, our proposed method is developed for nonholonomic agents under directed communications and can deal with general curves (i.e., both open and closed curves) and reject constant input disturbances.
The main contribution of our work is summarized as follows
\begin{itemize}
    \item [1)] A unified parametric representation of open and closed curves for the formation control problem.
    \item [2)] A leader-follower formation controller for nonholonomic agents to form parametric curves under directed communications and constant input disturbances.
    \item [3)] A rigorous stability analysis, numerical simulations, and experimental study to investigate the properties and validate the performance of the proposed method.
\end{itemize}

\emph{Notation.}
Matrices and vectors are denoted as bold letters, while scalars are denoted as plain letters.
We use $\mathbf{0}$ to denote a vector of appropriate dimension with all elements as zeros, and $\mathbf{0}_m$ to denote a square matrix of zeros.
$\mathbf{I}_n$ denotes a $n\times n$ identity matrix.
We use the symbol $\text{diag}(\cdot)$ to represent the diagonalization operation, where ``$\cdot$'' could be square matrices or scalars.

\emph{Organization.}
The rest of the paper is organized as follows:
Section II presents the mathematical preliminaries; Section III derives the formation controller; Section IV validates the theory with simulations and experiments;
Section V gives final conclusions.

\section{Mathematical Modeling}
\label{sec:Dynamics}
\subsection{Curve Representation}
Our goal in this paper is to study the multi-agent control method driving agents to achieve the formation along curves.
Therefore, the first step is to define the curve we expect the agents to form. 
We only consider the planar cases in this paper and assume that the desired curve can be parameterized by the following linear regression equation.
\begin{equation}
    \mathbf{c}=\mathbf{G}(s)\boldsymbol{\xi}
    \label{curve}
\end{equation}
where $\mathbf{c}=[c_x,c_y]^{\mathsf{T}}$ denotes the curve, $s\in[0,1]$ is the normalized length parameter, $\mathbf{G}(s)\in\mathbb{R}^{2\times2H}$ is the matrix containing basis functions dependent on the parameter $s$, $H$ is the number of basis functions, and $\boldsymbol{\xi}\in\mathbb{R}^{2H\times1}$ is the vector containing the coefficients.
By selecting different basis functions, we can represent the desired curve as different types, such as Fourier series, Bezier contours, and polynomial contours.

In real-world missions, we can obtain the curve coefficients by estimation.
First, we sample a series of points on the observed environmental curve, denoted as
\begin{equation}
    \mathbf{C}=[\mathbf{c}_1^{\mathsf{T}},\dots,\mathbf{c}_N^{\mathsf{T}}]^{\mathsf{T}}
\end{equation}
where $N$ is the total number of sampled points. Then, appropriately select a series of basis functions and compute the basis matrice corresponding to each sampled point, denoted as
\begin{equation}
    \mathbf{G}_h=[\mathbf{G}_1^{\mathsf{T}},\dots,\mathbf{G}_N^{\mathsf{T}}]^{\mathsf{T}}
\end{equation}
Finally, we can estimate the curve coefficients by
\begin{equation}
    \boldsymbol{\xi}=(\mathbf{G}_h^{\mathsf{T}}\mathbf{G}_h)^{-1}\mathbf{G}_h\mathbf{C}
    \label{coff est}
\end{equation}
Note that the number of sample points along the curve must satisfy $N>2H$ to guarantee the existence of $(\mathbf{G}_h^{\mathsf{T}}\mathbf{G}_h)^{-1}$, which is usually easy to fulfill in real missions.

\subsection{Dynamic Model}
The agents considered in this work are assumed to conduct planar motions and have nonholonomic dynamics with input disturbance.
We define the configuration of agent $i$ as $[\mathbf{x}_i^{\mathsf{T}}, \theta_i]^{\mathsf{T}}$, where $\mathbf{x}_i=[x_i, y_i]^{\mathsf{T}}$ and $\theta_i$ are the center position and the orientation of the agent, respectively.
The dynamic model of agent $i$ is given by
\begin{equation}
    \begin{bmatrix}
    \dot{x}_i \\ \dot{y}_i \\ \dot{\theta}_i
    \end{bmatrix}=
    \begin{bmatrix}
    \cos{\theta_i} & 0 \\
    \sin{\theta_i} & 0 \\
    0 & 1
    \end{bmatrix}
    (\mathbf{u}_i+\mathbf{d}_i)
    \label{dynamics}
\end{equation}
where $\mathbf{u}_i=[v_i, \omega_i]^{\mathsf{T}}$ denotes the control input, with $v_i$ and $\omega_i$ as the linear and angular velocities of the agent, respectively; $\mathbf{d}_i=[d_{i1}, d_{i2}]^{\mathsf{T}}$ denotes the constant input disturbance.

To simplify the controller design and stability analysis, we define a virtual control point adopting the input/output feedback linearization\cite{siliciano2010robotics}.
The virtual control point is given by the following change of coordinates \cite{Becerraelal} (see also \cite{ROMROD}):
\begin{equation}
    \left\{
    \begin{aligned}
    & \bar{x}_{i}:=x_i+\ell\cos{\theta_i}\\
    & \bar{y}_{i}:=y_i+\ell\sin{\theta_i}\\
    & \bar{\theta}_i:=\theta_i
    \end{aligned}
    \right.
    \label{change_of_coord}
\end{equation}
where $\ell\neq0$ is an arbitrary scalar that translates the agent's position to an arbitrary close location.

By using these new coordinates, we can obtain the agent's shifted position $\bar{\mathbf{x}}_i=[\bar{x}_i, \bar{y}_i]^{\mathsf{T}}$, whose time derivative yields a modified dynamic model of the form
\begin{equation}
    \left\{
    \begin{aligned}
    & \dot{\bar{\mathbf{x}}}_i=\mathbf{R}_i(\theta_i)(\mathbf{u}_i+\mathbf{d}_i)\\
    & \dot{\bar{\theta}}_i=\omega_i + d_{i2}
    \end{aligned}
    \right.
    \label{modified_dynamics}
\end{equation}
for a full-rank matrix defined as:
\begin{equation}
    \mathbf{R}_i(\theta_i)=
    \begin{bmatrix}
    \cos{\theta_i} & -\ell\sin{\theta_i}\\
    \sin{\theta_i} & \ell\cos{\theta_i}
    \end{bmatrix}
    \label{R}
\end{equation}

By defining a new control input $\bar{\mathbf{u}}_i=\mathbf{R}_i(\theta_i)\mathbf{u}_i$ and a new input disturbance $\bar{\mathbf{d}}_i=\mathbf{R}_i(\theta_i)\mathbf{d}_i$, the position dynamics in \eqref{modified_dynamics} can be reduced to a single integrator with input disturbance 
\begin{equation}
    \dot{\bar{\mathbf{x}}}_i=\bar{\mathbf{u}}_i+\bar{\mathbf{d}}_i
    \label{position dynamics}
\end{equation}

\subsection{Interaction Topology}
The interactions topology among the agents can be represented by a graph $\mathcal{G}=(\mathcal{G}, \mathcal{E}, a_{ij})$, where $\mathcal{V}=\{1,\dots,n\}$ denotes the set of nodes, $\mathcal{E}=\{(i,j)\in\mathcal{V}\times\mathcal{V}:a_{ij}\neq0\}$ denotes the set of edges, and $a_{ij}>0$ denotes the weight of interactions with $ij$ indicating the $i$-th and $j$-th agents.
The Laplacian matrix $\mathbf{L}=[l_{ij}]\in\mathbb{R}^{n\times n}$ of the graph is defined by
\begin{equation}
\label{laplacian}
l_{ij} = 
\begin{cases}
\sum_{i=1}^{n}a_{ij}, & \text{if}\ {i = j}\\
- {a_{ij}}, & \text{if}\ {i \ne j}
\end{cases}
\end{equation}
for $l_{ij}$ as the entry of $\mathbf{L}$ at the $i$-th row and $j$-th column\cite{mesbahi2010graph}.

In this work, we assume that the interaction topology of agents can be represented by a directed graph containing a rooted spanning tree, i.e., there is a node that can be connected to all the remaining nodes of the graph.
The Laplacian has a simple eigenvalue 0 with $\textbf{1}_n\in\mathbb{R}^{n\times1}$ as its corresponding right eigenvector, where $\textbf{1}_n$ denotes a vector with all elements equal to one.
All the other eigenvalues of the Laplacian have positive real parts\cite{panteley2020strict}.

Besides, we can define an augmented graph for a leader-follower multi-agent system on the basis of the graph of the followers.
We can describe the communication between the leader and the followers by a nonnegative matrix $\mathbf{B}=\text{diag}(b_i)$, where $b_i>0$ if and only if there is an edge between the $i$-th follower and the leader.
Assuming that the augmented graph contains a spanning tree, we have the following theorem on the construction of Lyapunov functions for stability analysis.

\begin{theorem}(Zhang \emph{et. al.}\cite{zhang2015constructing})
    Let
    \begin{equation}
        \left\{
        \begin{aligned}
            & \mathbf{q}=[q_i]^{\mathsf{T}}=(\mathbf{L}+\mathbf{B})^{-1}\mathbf{1}_n\\
            & \mathbf{p}=[p_i]^{\mathsf{T}}=(\mathbf{L}+\mathbf{B})^{-\mathsf{T}}\mathbf{1}_n\\
            & \mathbf{P}=\text{diag}(p_i/q_i)\\
            & \mathbf{Q} = \mathbf{P}(\mathbf{L}+\mathbf{B})+(\mathbf{L}+\mathbf{B})^{\mathsf{T}}\mathbf{P}
        \end{aligned}
        \right.
    \end{equation}
    Then both $\mathbf{P}$ and $\mathbf{Q}$ are positive definite.
    \label{theorem1}
\end{theorem}

\section{Controller Design}
\label{sec:Contrnoller}
\subsection{Controller Design}
By stacking the position dynamics \eqref{position dynamics} of agents, we can represent the position dynamics of the multi-agent system as
\begin{equation}
    \dot{\bar{\mathbf{x}}}=\bar{\mathbf{u}}+\bar{\mathbf{d}}
    \label{stacked position dynamics}
\end{equation}
where $\bar{\mathbf{x}}=[\bar{\mathbf{x}}_i^{\mathsf{T}}]^{\mathsf{T}}$, $\bar{\mathbf{u}}=[\bar{\mathbf{u}}_i^{\mathsf{T}}]^{\mathsf{T}}$, and $\bar{\mathbf{d}}=[\bar{\mathbf{d}}_i^{\mathsf{T}}]^{\mathsf{T}}\in\mathbb{R}^{2n\times1}$ are the extended position vector, input vector, and disturbance vector, respectively.

We uniformly assign length parameters to each agent to ensure the equal arc length separation of agents along the curve.
For that, we set a sequence of length parameters corresponding to the agents $\mathbf{s}=[s_1,\dots,s_i,\dots,s_n]^{\mathsf{T}}$, where $s_i$ is the length parameter corresponding to the $i$th agent, calculated by $s_i=(i-1)/n$. 
With that, we can calculate the matrix of basis functions corresponding to each agent as
\begin{equation}
    \bar{\mathbf{G}}=[\mathbf{G}_1(s_1)^{\mathsf{T}},\dots,\mathbf{G}_n(s_n)^{\mathsf{T}}]^{\mathsf{T}}
\end{equation}

Now, we can define the agent's position errors as
\begin{equation}
    \bar{\mathbf{x}}_e=\bar{\mathbf{x}}-\bar{\mathbf{G}}\boldsymbol{\xi}
    \label{pos err}
\end{equation}
and the curve coefficient errors as
\begin{equation}
    \boldsymbol{\xi}_e=\bar{\mathbf{G}}^{+}\bar{\mathbf{x}}-\boldsymbol{\xi}
    \label{coff err}
\end{equation}
where $\bar{\mathbf{G}}^{+}$ is the pseudoinverse of $\bar{\mathbf{G}}$.

To achieve the formation control and reject the input disturbance, we propose the following leader-follower controller
\begin{equation}
    \bar{\mathbf{u}}_i=
    \begin{cases}
        -k_1(\bar{\mathbf{x}}_i-\mathbf{G}_i\boldsymbol{\xi})-k_2\mathbf{R}_i\hat{\boldsymbol{\delta}}_i & \text{leader}\\
        -k_1\sum_{j\in\mathcal{N}_i}a_{ij}(\mathbf{G}_i-\mathbf{G}_j)\boldsymbol{\xi}_e-k_2\mathbf{R}_i\hat{\boldsymbol{\delta}}_i & \text{follower}
    \end{cases}
    \label{controller}
\end{equation}
where $k_1,k_2>0$ are the control gains, $\hat{\boldsymbol{\delta}}_i$ denotes the estimation of the input disturbance, and $\mathcal{N}_i=\{j\in\mathcal{V}:a_{ij}\neq0\}$ is the set of neighbors of agent $i$. The first part of the controller drives agents to achieve the formation; The second part of the controller eliminates the input disturbance. 
The disturbance estimation is updated by
\begin{equation}
    \dot{\hat{\boldsymbol{\delta}}}_i=k_2\mathbf{R}_i^{\mathsf{T}}(\bar{\mathbf{x}}_i-\mathbf{G}_i\boldsymbol{\xi})
    \label{estimation update}
\end{equation}

\subsection{Stability Analysis}
Without loss of generality, we assume that agent 1 is the root of the spanning tree contained in the interaction topology and the leader of the multi-agent system.
This indicates that the first row of the Laplacian $\mathbf{L}$ is a zero vector.
Then, we can stack the controller \eqref{controller} for each agent to obtain the control input $\bar{\mathbf{u}}$ for the whole multi-agent system.
\begin{equation}
    \bar{\mathbf{u}}=-k_1(\bar{\mathbf{L}}\bar{\mathbf{G}}\boldsymbol{\xi}_e+\bar{\boldsymbol{\Lambda}}\bar{\mathbf{x}}_e)-k_2\mathbf{R}\hat{\boldsymbol{\delta}}
    \label{control input}
\end{equation}
where $\bar{\mathbf{L}}=\mathbf{L}\otimes\mathbf{I}_2$ is the extended Laplacian matrix with $\otimes$ representing the Kronecker product; $\bar{\boldsymbol{\Lambda}}=\boldsymbol{\Lambda}\otimes\mathbf{I}_2\in\mathbb{R}^{2n\times2n}$ for $\boldsymbol{\Lambda}=\text{diag}(1, \mathbf{0}_{n-1})$; $\mathbf{R}=\text{diag}(\mathbf{R}_i)\in\mathbb{R}^{2n\times2n}$;  and $\hat{\boldsymbol{\delta}}=[\hat{\boldsymbol{\delta}}_i^{\mathsf{T}}]^{\mathsf{T}}\in\mathbb{R}^{2n\times1}$.

Replacing the controller \eqref{control input} into the stacked dynamics \eqref{whole dynamics}, we can obtain the complete dynamics of the multi-agent system
\begin{equation}
    \dot{\bar{\mathbf{x}}}=-k_1(\bar{\mathbf{L}}\bar{\mathbf{G}}\boldsymbol{\xi}_e+\bar{\boldsymbol{\Lambda}}\bar{\mathbf{x}}_e)-k_2\mathbf{R}\hat{\boldsymbol{\delta}}+\bar{\mathbf{d}}
    \label{whole dynamics}
\end{equation}

Now, we define a disturbance estimation error measurement
\begin{equation}
    \tilde{\boldsymbol{\delta}}=\hat{\boldsymbol{\delta}}-\frac{1}{k_2}\mathbf{d}
    \label{dist err}
\end{equation} 
for ${\bf d}=[{\bf d}_i^{\mathsf{T}}]^{\mathsf{T}} \in \mathbb{R}^{2n\times1}$.
Then, we can rewrite \eqref{whole dynamics} as
\begin{equation}
    \dot{\bar{\mathbf{x}}}=-k_1(\bar{\mathbf{L}}\bar{\mathbf{G}}\boldsymbol{\xi}_e+\bar{\boldsymbol{\Lambda}}\bar{\mathbf{x}}_e)-k_2\mathbf{R}\tilde{\boldsymbol{\delta}}
    \label{whole dynamics new}
\end{equation}
It can be derived $\dot{\tilde{\boldsymbol{\delta}}}=\dot{\hat{\boldsymbol{\delta}}}=k_2\mathbf{R}^{\mathsf{T}}\bar{\mathbf{x}}_e$ considering that $\mathbf{d}$ is constant.

The stability analysis of the proposed controller is based on the following assumption.
\begin{assumption}
The rank of matrix $\bar{\mathbf{G}}\in\mathbb{R}^{2n\times2H}$  always satisfies $\text{rank}(\bar{\mathbf{G}})=\min\{2n, 2H\}$. 
\label{assum1}
\end{assumption}

\begin{assumption}
    The number of agents and basis functions satisfies $n\leq H$.
    \label{assum2}
\end{assumption}

With Assumption \ref{assum1}, we can ensure the existence of the pseudoinverse $\bar{\mathbf{G}}^+$.
Specifically, we can compute $\bar{\mathbf{G}}^+$ by $\bar{\mathbf{G}}^+=(\bar{\mathbf{G}}^{\mathsf{T}}\bar{\mathbf{G}})^{-1}\bar{\mathbf{G}}^{\mathsf{T}}$ when $n\geq H$ and $\bar{\mathbf{G}}^+=\bar{\mathbf{G}}^{\mathsf{T}}(\bar{\mathbf{G}}\bar{\mathbf{G}}^{\mathsf{T}})^{-1}$ when $n<H$.
With Assumption \ref{assum2}, we can obtain $\bar{\mathbf{G}}\bar{\mathbf{G}}^+=\mathbf{I}_{2n}$.
Then, we can derive that $\bar{\mathbf{G}}\boldsymbol{\xi}_e=\bar{\mathbf{G}}\bar{\mathbf{G}}^+\bar{\mathbf{x}}_e=\bar{\mathbf{x}}_e$

\begin{remark}
    The condition proposed in Assumption \ref{assum2} is reasonable.
    We usually need more basis functions to accurately approximate a curve with a complex shape structure.
    As for ``simple" curves, such as straight lines and circles, we can also regard them as a combination of many basis functions with the coefficients of high-order basis functions as zeros.
    For example, we can represent a simple straight line as
    \begin{equation}
        \mathbf{c}=\mathbf{a}+\mathbf{b}s+\mathbf{0}(s^2+s^3+\dots)
    \end{equation}
    for $\mathbf{a}$ and $\mathbf{b}\in\mathbb{R}^{2\times1}$ as coefficients.
\end{remark}

\begin{proposition}
    Consider a multi-agent system of $n$ agents with dynamics \eqref{position dynamics} and a directed graph containing a rooted spanning tree.
    Given a \emph{static} planar curve with the form of \eqref{curve}, the controller \eqref{control input} drives the agents to be uniformly distributed on the curve and ensures the asymptotic stability of the position errors $\bar{\mathbf{x}}_e$.
    The orientations of agents will converge to some constant values.
\end{proposition}

\begin{proof}
    It can be seen that the Laplacian $\mathbf{L}$ and the matrix $\boldsymbol{\Lambda}$ satisfies Theorem \ref{theorem1}.
    Therefore, replace $\mathbf{B}$ with $\boldsymbol{\Lambda}$ in Theorem \ref{theorem1}, then we can obtain positive definite matrice $\mathbf{P}$, $\mathbf{Q}\in\mathbb{R}^{n\times n}$.

    Consider the following Lyapunov function
    \begin{equation}
        V= \bar{\mathbf{x}}_e^{\mathsf{T}}(\mathbf{P}\otimes\mathbf{I}_2)\bar{\mathbf{x}}_e+\tilde{\boldsymbol{\delta}}^{\mathsf{T}}(\mathbf{P}\otimes\mathbf{I}_2)\tilde{\boldsymbol{\delta}}
        \label{lyap}
    \end{equation}
    Compute the time derivative of \eqref{lyap}, then we obtain
    \begin{equation}
        \begin{aligned}
            \dot{V}
            & = 2\bar{\mathbf{x}}_e^{\mathsf{T}}(\mathbf{P}\otimes\mathbf{I}_2)\dot{\bar{\mathbf{x}}}_e+2\tilde{\boldsymbol{\delta}}^{\mathsf{T}}(\mathbf{P}\otimes\mathbf{I}_2)\dot{\tilde{\boldsymbol{\delta}}}\\
            & = 2\bar{\mathbf{x}}_e^{\mathsf{T}}(\mathbf{P}\otimes\mathbf{I}_2)[-k_1(\bar{\mathbf{L}}\bar{\mathbf{G}}\boldsymbol{\xi}_e+\bar{\boldsymbol{\Lambda}}\bar{\mathbf{x}}_e)-k_2\mathbf{R}\tilde{\boldsymbol{\delta}}]\\
            & \quad+2\tilde{\boldsymbol{\delta}}^{\mathsf{T}}(\mathbf{P}\otimes\mathbf{I}_2)k_2\mathbf{R}^{\mathsf{T}}\bar{\mathbf{x}}_e\\
            & = -2k_1\bar{\mathbf{x}}_e^{\mathsf{T}}(\mathbf{P}\otimes\mathbf{I}_2)(\mathbf{L}\otimes\mathbf{I_2}+\boldsymbol{\Lambda}\otimes\mathbf{I_2})\bar{\mathbf{x}}_e\\
            & \quad-2k_2\bar{\mathbf{x}}_e^{\mathsf{T}}[(\mathbf{P}\otimes\mathbf{I}_2)\mathbf{R}-\mathbf{R}(\mathbf{P}\otimes\mathbf{I}_2)]\tilde{\boldsymbol{\delta}}
        \end{aligned}
    \end{equation}
    
    Checking the structure of $\mathbf{R}$ and $\mathbf{P}\otimes\mathbf{I}_2$, we can derive that
    \begin{equation}
        (\mathbf{P}\otimes\mathbf{I}_2)\mathbf{R}=\text{diag}\left(\frac{p_i}{q_i}\mathbf{I_2}\mathbf{R}_i\right)=\mathbf{R}(\mathbf{P}\otimes\mathbf{I}_2)
    \end{equation}
    Therefore, we have 
    \begin{equation}
        \dot{V}=-\bar{\mathbf{x}}_e^{\mathsf{T}}(\mathbf{Q}\otimes\mathbf{I}_2)\bar{\mathbf{x}}_e\leq0
    \end{equation}
    where $\dot{V}=0$ if and only if $\bar{\mathbf{x}}_e=\mathbf{0}$.
    Therefore, $\bar{\mathbf{x}}_e=\mathbf{0}$ is asymptotically stable.
    Furthermore, from $\dot{\bar{\mathbf{x}}}_e$ and knowing that $\bar{\mathbf{x}}_e=\mathbf{0}$, we have $\mathbf{R}\Tilde{\boldsymbol{\delta}}=\mathbf{0}$ and $\dot{\tilde{\boldsymbol{\delta}}}=\mathbf{0}$.
    Since $\mathbf{R}$ is a full-rank matrix, we can derive that $\tilde{\boldsymbol{\delta}}=\mathbf{0}$ is also asymptotically stable.

    Now, we consider the evolution of the orientations of the agents. Since we have proved the asymptotic stability of $\bar{\mathbf{x}}_e=\mathbf{0}$ and $\tilde{\boldsymbol{\delta}}=\mathbf{0}$, we can derive $\bar{\mathbf{u}}\rightarrow\mathbf{-Rd}$.
    Considering $\mathbf{u}=\mathbf{R}^{-1}\bar{\mathbf{u}}$ and $\mathbf{u}_i=[v_i,\omega_i]^{\mathsf{T}}$, we can obtain
    \begin{equation}
        \dot{\theta}_i=\omega_i+d_{i2}\rightarrow0
    \end{equation}
    which indicates the zero dynamics of agents' orientation.
    Therefore, we can derive that the orientations of agents will converge to some bounded constant solutions.
\end{proof}

\section{Results}

\begin{figure}[t]
    \centering
    \includegraphics[width=\columnwidth]{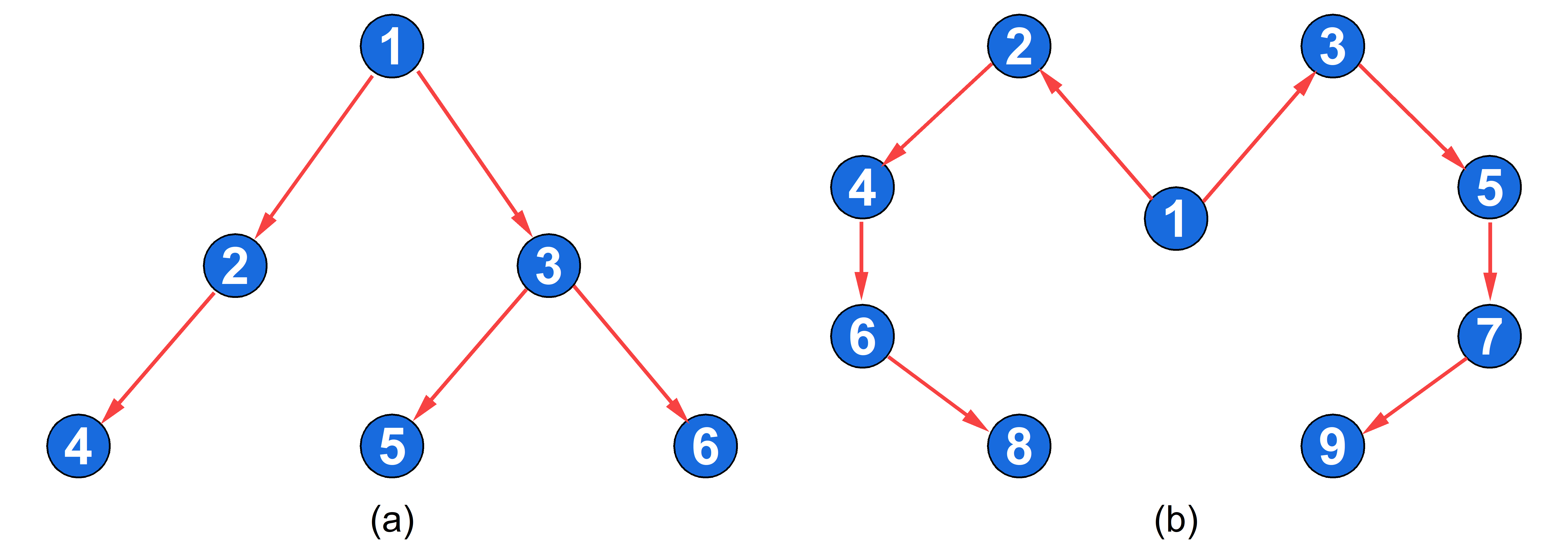}
    \caption{Interaction graph of agents.}
    \label{fig:graph}
\end{figure}

\subsection{Test with Single Curves}
We first test our method by driving agents to form a single static curve.
The simulations are conducted for a multi-agent system with six agents ($n=6$) interacting through a directed graph shown in Fig. \ref{fig:graph}(a).
We set the weights of all edges as one, the scalar for the change of coordinates as $\ell=0.01$, and the control gains as $k_1=k_2=1$. 
The constant input disturbances of the agents are set to $\mathbf{d}_i=[1,1]^{\mathsf{T}}$.
Agents depart from random initial positions with random initial orientations.

\subsubsection{Closed Curve} 
\label{closed}
The first simulation is conducted for a closed curve of the following form
\begin{equation}
    \begin{aligned}
        & c_x = (8+\sin{4\pi s})\cos{2\pi s}+4\\
        & c_y = (8+\cos{4\pi s})\sin{2\pi s}+4
    \end{aligned}
    \label{curve1}
\end{equation}
We approximate \eqref{curve1} to a truncated Fourier series with six harmonics (i.e., $H=13$, and Assumption \ref{assum2} is satisfied.) by \eqref{coff est}.
The matrix $\mathbf{G}_i(s_i)$ has a structure of $\mathbf{G}_i(s_i)=[\mathbf{g}_1,\dots,\mathbf{g}_6,\mathbf{I}_2]\in\mathbb{R}^{2\times26}$, where
\begin{equation}
    \mathbf{g}_h=
    \begin{bmatrix}
        \cos{2\pi hs_i} & \sin{2\pi hs_i} & 0 & 0\\
        0 & 0 & \cos{2\pi hs_i} & \sin{2\pi hs_i}
    \end{bmatrix}
\end{equation}
for $h=1,\dots, 6$.
It can be checked that $\bar{\mathbf{G}}=[\mathbf{G}_i(s_i)^{\mathsf{T}}]^{\mathsf{T}}$ satisfies Assumption \ref{assum1}.

\begin{figure}[t]
    \centering
    \includegraphics[width=\columnwidth]{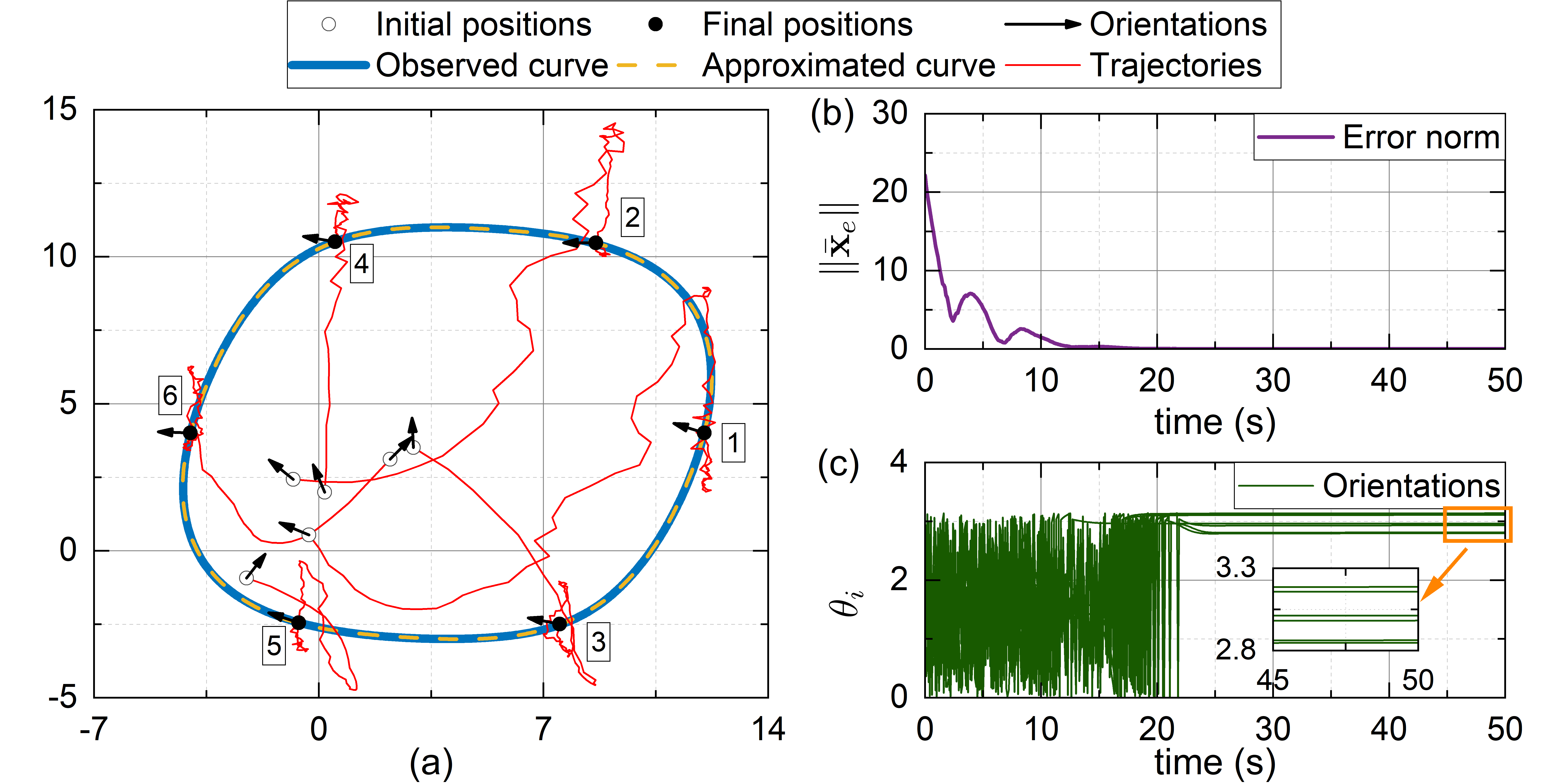}
    \caption{Results of forming a closed curve.}
    \label{fig:case1}
\end{figure}

Running the simulation, we obtain the results shown in Fig. \ref{fig:case1}(a).
We can see that agents depart from random positions and gradually converge to the desired curve.
We record the norm of position errors and the evolution of agents' orientation (rounded to $[0, 2\pi]$) in Fig. \ref{fig:case1}(b) and (c), respectively.
It can be seen that the errors asymptotically converge to zero and agents' orientations finally reach some constant values as we predicted in the stability analysis.

\subsubsection{Open Curve}
The second simulation is conducted for an open curve of the following form
\begin{equation}
    \mathbf{c}=(1-s)^3\mathbf{o}_1+3s(1-s)^2\mathbf{o}_2+3s^2(1-s)\mathbf{o}_3+s^3\mathbf{o}_4
    \label{curve2}
\end{equation}
for $\mathbf{o}_1=[3.5, 3]^{\mathsf{T}}$, $\mathbf{o}_2=[-0.5, -4]^{\mathsf{T}}$, $\mathbf{o}_3=[-2, 6]^{\mathsf{T}}$, and $\mathbf{o}_4=[-2, -1]^{\mathsf{T}}$.
We approximate \eqref{curve2} to a sixth-order polynomial (i.e., $H=7$, and Assumption \ref{assum2} is satisfied). The matrix $\mathbf{G}_i(s_i)$ has a structure of $\mathbf{G}_i(s_i)=\mathbf{g}_i\otimes\mathbf{I}_2\in\mathbb{R}^{2\times14}$ for $\mathbf{g}_i=[1,s_i,\dots, s_i^6]$.
It can be checked that $\bar{\mathbf{G}}=[\mathbf{G}_i(s_i)^{\mathsf{T}}]^{\mathsf{T}}$ satisfies Assumption \ref{assum1}.

\begin{figure}[t]
    \centering
    \includegraphics[width=\columnwidth]{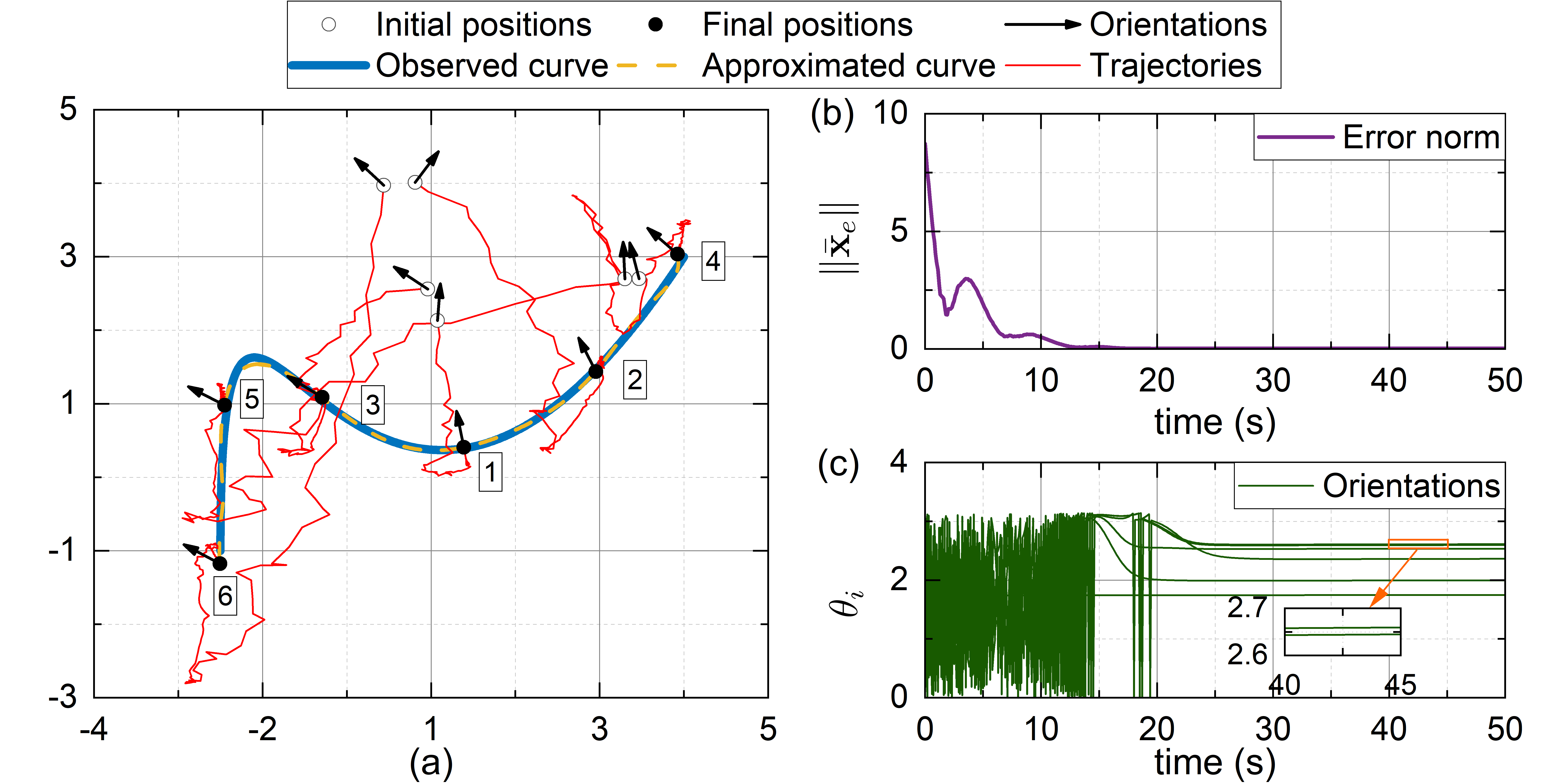}
    \caption{Results of forming an open curve.}
    \label{fig:case2}
\end{figure}

Running the simulation, we obtain the results shown in Fig. \ref{fig:case2}.
Similarly to the first simulation, agents depart from random initial positions and gradually converge to the desired curve (see Fig. \ref{fig:case2}(a)).
The position errors asymptotically converge to zero (see Fig. \ref{fig:case2}(b)) and agents finally reach some constant orientations (see Fig. \ref{fig:case2}(c)) as predicted in the stability analysis.

\subsection{Test with Curve Shape Changes}
\label{shape_change}
In real-world missions, the curve that agents need to form could vary between different patterns, and our method also has the potential to handle this case.
In this simulation, we simulate the situation where agents need to form one curve first and then move to another curve.
We adopt a multi-agent system of nine agents ($n=9$) interacting through a directed graph shown in Fig. \ref{fig:graph}(b). 
We set the control gain to $k_1=1,k_2=0.75$ and keep other settings the same as in the previous simulations.

The simulation is conducted for a closed curve of the following form
\begin{equation}
    \begin{aligned}
        & c_x = (8+2\sin{4\pi s})\cos{2\pi s}-12\\
        & c_y = (8+2\cos{4\pi s})\sin{2\pi s}+4 
    \end{aligned}
    \label{curve3}
\end{equation}
for the first 75 seconds and switched to another closed curve of the following form
\begin{equation}
    \begin{aligned}
        & c_x = (8+2\cos{2\pi s}+\sin{4\pi s})\cos{2\pi s}+24\\
        & c_y = (8+2\cos{2\pi s}+\sin{4\pi s})\sin{2\pi s}+4
    \end{aligned}
    \label{curve4}
\end{equation}
after 75 seconds.
We approximate both curves to a truncated Fourier series with eight harmonics (i.e. $H=17$). 
Similarly to Section \ref{closed}, we can check that both Assumptions \ref{assum1} and \ref{assum2} are satisfied.

\begin{figure}[t]
    \centering
    \includegraphics[width=\columnwidth]{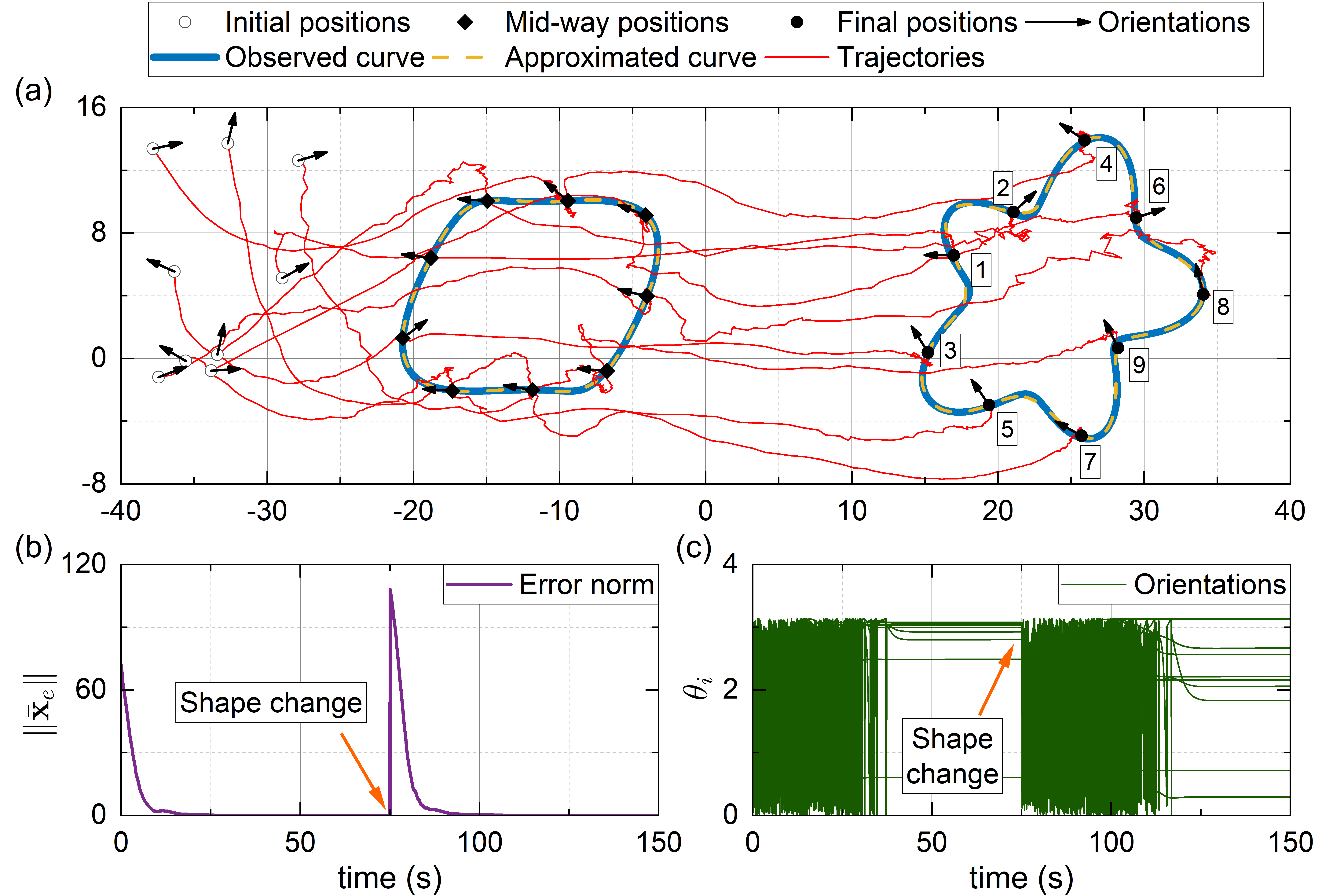}
    \caption{Results of the simulation with curve shape changes.}
    \label{fig:case3}
\end{figure}

Running the simulation, we can obtain the results shown in Fig. \ref{fig:case3}(a).
Agents converge to the first desired curve and then to the second desired curve.
This can also be seen in Fig. \ref{fig:case3}(b). 
The norm of the position errors asymptotically converges to zero when agents reach the first curve; then jumps up when the second desired curve appears and asymptotically converges to zero again when agents reach the second curve.
The orientations of agents (see Fig. \ref{fig:case3}(c)) reach some constant values as predicted in the stability analysis every time when agents reach a desired curve. 

\begin{remark}
    In real-world missions, the desired curves could be environmental boundaries or the perimeters of objects, etc., that can be observed by appropriate sensors, such as cameras, and extracted from videos or images.
    However, our core problem in this study is developing a control method for forming those curves.
    Therefore, we simply predefine some desired curves to pretend that we have already obtained the desired curves from observation for the convenience of simulations and call them the observed curves.
\end{remark}

\subsection{Experiments}
\begin{figure}[t]
    \centering
    \includegraphics[width=0.9\columnwidth]{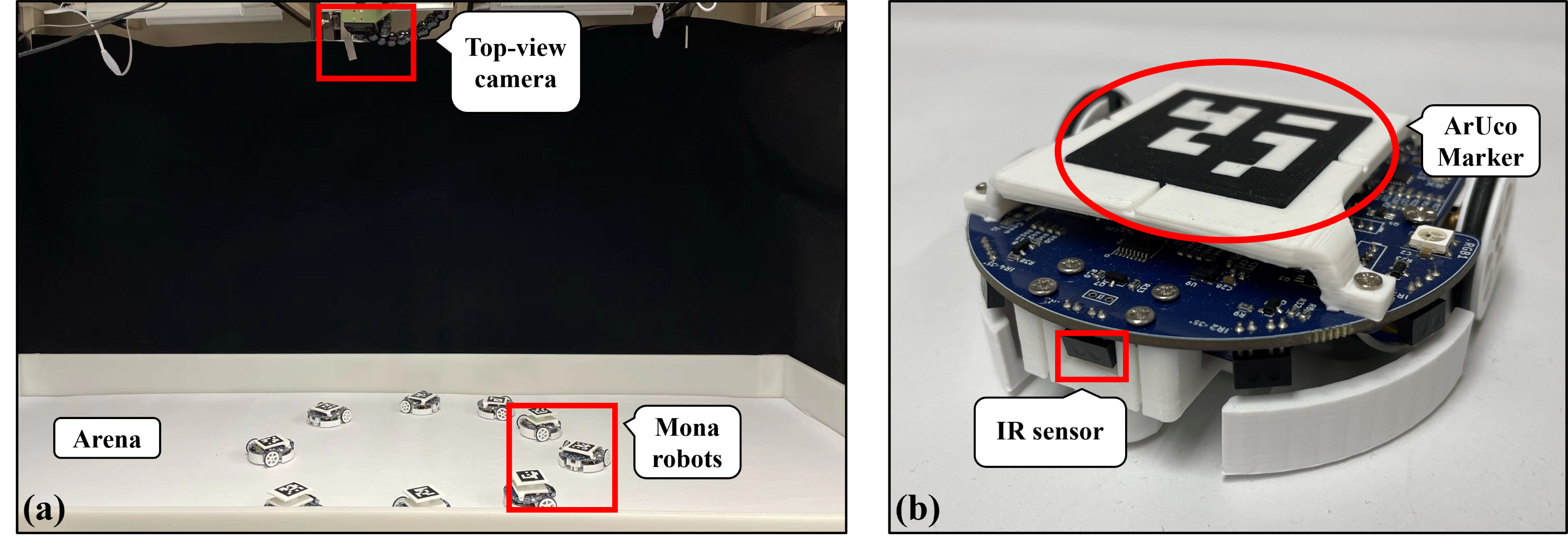}
    \caption{The configuration of the experimental platform.}
    \label{fig:platform}
\end{figure}
To verify the performance of our method in real-world tasks, we conduct a simple experimental test with the platform shown in Fig. \ref{fig:platform}.
The platform consists of nine ($n=9$) Mona robots \cite{arvin2019mona} with nonholonomic dynamics, a top-view camera to observe the pose of the robots, a control PC to collect data and send motion commands, and a 1.8m$\times$0.8m arena for robots to move around.
We build a Wi-Fi communication at a rate of 10 Hz between the control PC and robots based on ROS.

During the experiment, robots interact with each other through the topology shown in Fig. \ref{fig:graph}(b).
We empirically set the gains of the controller to $k_1=\frac{2}{3}$ and $k_2=\frac{1}{3}$.
We apply a constant input disturbance $\mathbf{d}_i=[0.5, 0.5]^{\mathsf{T}}$ cm/s to each robot through the control PC.
Considering the hardware constraints of the robots, we limit the linear velocity applied to the robot wheels to $\pm2.25$ cm/s.

The experiment is conducted for a closed curve of the following form
\begin{equation}
    \begin{aligned}
        & c_x = 0.225(1-\sin{2\pi s})\cos{2\pi s}+0.75\\
        & c_y = 0.225(1-\sin{2\pi s})\sin{2\pi s}+0.675
    \end{aligned}
\end{equation}
for the first 60 seconds and switched to another closed curve of the following form
\begin{equation}
    \begin{aligned}
        & c_x = (0.03\sin{4\pi s}+0.06\cos{10\pi s}+0.225)\cos{2\pi s}+1.125\\
        & c_y = (0.03\sin{4\pi s}+0.06\cos{10\pi s}+0.225)\sin{2\pi s}+0.375
    \end{aligned}
\end{equation}
after 60 seconds. The approximations of both curves follow the same operations as in Section \ref{shape_change}.

The results of the experiment are shown in Fig. \ref{fig:experi}(a), where several screenshots show that the robots start from random initial positions and finally form the desired curves.
We record the position error norm and orientations in Fig. \ref{fig:experi}(b) and (c), respectively.
The convergence of error norm and orientations can be seen every time the robots reach a desired curve.
The entire experiment process can be seen in the supplementary video at \url{https://vimeo.com/872927697}.

\begin{figure}[t]
    \centering
    \includegraphics[width=\columnwidth]{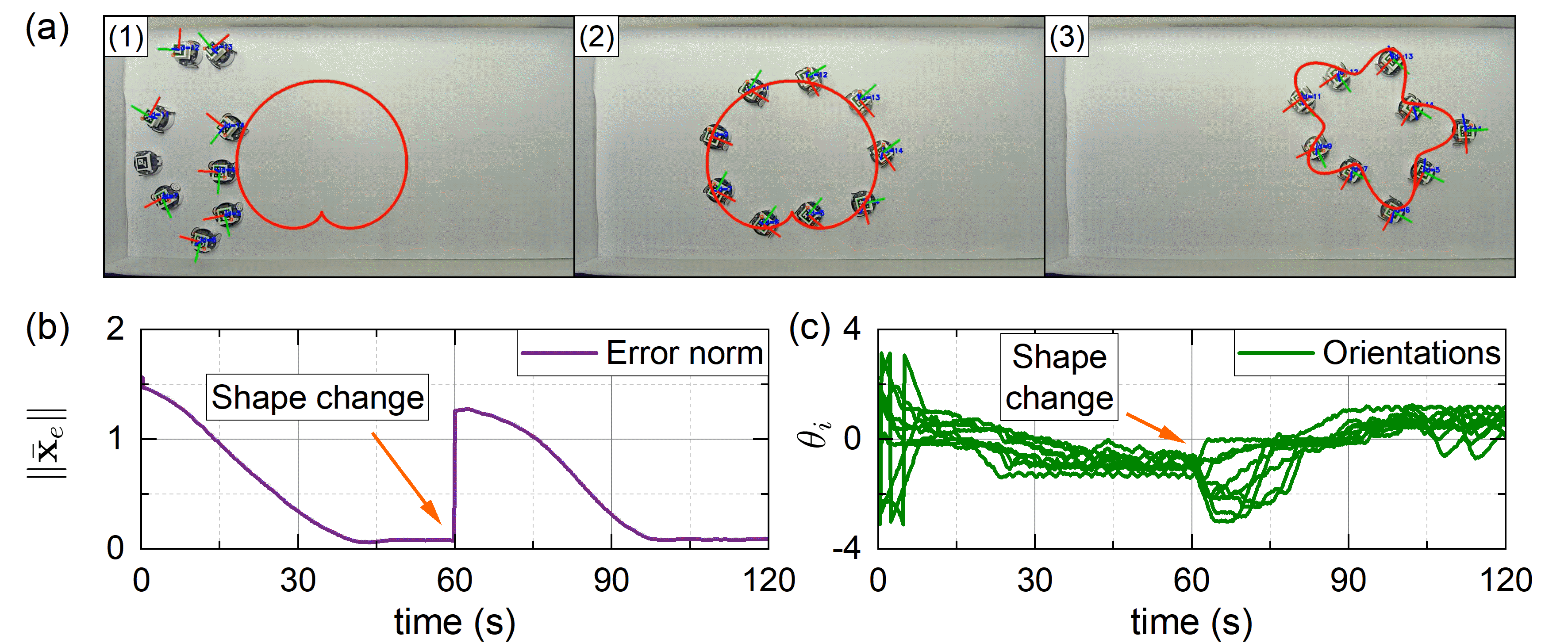}
    \caption{Results of the experiment with curve shape changes.}
    \label{fig:experi}
\end{figure}

\section{Conclusion}
In this work, we first present a unified parametric representation for general curves, that is, open and closed curves.
Then, we derive a leader-follower formation controller based on the parametric equations to form the desired curves under directed communications and constant input disturbances.
We propose a Lyapunov-based stability analysis to prove the asymptotic stability of the proposed controller.
Numerical simulations and experimental studies show the desirable performance of our proposed method in dealing with different cases, such as open curves, closed curves, and curve shape changes.
The method is with high scalability and has the potential to be applied to various real-world missions.
However, there are still some limits to the proposed method, such as dealing with continuously evolving curves and working with time delays.
In the future, we will keep extending the proposed method to solve the aforementioned problems.
 

\end{document}